\documentclass[11pt]{amsart}

\usepackage[margin=2.5cm]{geometry}
\usepackage{tikz}
\tikzset{node distance=2cm, auto}
\usetikzlibrary{matrix,arrows,decorations.pathmorphing}

\usepackage{graphicx}
\usepackage{amsmath,amssymb}

\usepackage{amscd}
\usepackage{verbatim}
\usepackage{enumerate}
\usepackage{amssymb,amsfonts}
\usepackage[all,arc]{xy}
\usepackage{enumerate}
\usepackage{mathrsfs}
\usepackage{bbm}
\usepackage{amsmath}           
\usepackage{amssymb}           
\usepackage{latexsym}          
\usepackage{mathrsfs}          
\usepackage{eucal}             
\usepackage{amsthm}            
\usepackage{mathtools}         
\usepackage{stmaryrd}          
\usepackage{esint}             
\usepackage{extarrows}         
\usepackage{enumerate}         
\usepackage{xcolor}            
\usepackage{graphicx}          
\usepackage{floatrow}
\usepackage{tikz}              
\usepackage{tikz-3dplot}       
\usepackage{tikz-cd}           
\usepackage{hyperref}  

\def\makeautorefname#1#2{\expandafter\def\csname#1autorefname\endcsname{#2}}

\def\equationautorefname~#1\null{(#1)\null}
\makeautorefname{footnote}{footnote}%
\makeautorefname{item}{item}%
\makeautorefname{figure}{Figure}%
\makeautorefname{table}{Table}%
\makeautorefname{part}{Part}%
\makeautorefname{appendix}{Appendix}%
\makeautorefname{chapter}{Chapter}%
\makeautorefname{section}{Section}%
\makeautorefname{subsection}{Section}%
\makeautorefname{subsubsection}{Section}%
\makeautorefname{theorem}{Theorem}%
\makeautorefname{thm}{Theorem}%
\makeautorefname{sta}{Statement}%
\makeautorefname{cor}{Corollary}%
\makeautorefname{lem}{Lemma}%
\makeautorefname{prop}{Proposition}%
\makeautorefname{pro}{Property}
\makeautorefname{conj}{Conjecture}%
\makeautorefname{conj}{Convention}%
\makeautorefname{defn}{Definition}%
\makeautorefname{notn}{Notation}
\makeautorefname{notns}{Notations}
\makeautorefname{rem}{Remark}%
\makeautorefname{rems}{Remarks}%
\makeautorefname{quest}{Question}%
\makeautorefname{exmp}{Example}%
\makeautorefname{ax}{Axiom}%
\makeautorefname{claim}{Claim}%
\makeautorefname{ass}{Assumption}%
\makeautorefname{asses}{Assumptions}%
\makeautorefname{con}{Construction}%
\makeautorefname{prob}{Problem}%
\makeautorefname{warn}{Warning}%
\makeautorefname{obs}{Observation}%

\newtheorem{thm}{Theorem}[section]

\newtheorem{prop}{Proposition}[section]

\theoremstyle{definition}
\newtheorem{defn}{Definition}[section]

\newtheorem{exmp}{Example}[section]

\newtheorem{rem}{Remark}[section]

\DeclareMathOperator{\Ad}{Ad}

\DeclareMathOperator{\Diff}{Diff}

\title{Spin Toda systems and N-point spherical Whittaker functions}

\author{Haoyi Li}
\address{H.L.: Qiuzhen College, Tsinghua University, Beijing}
\email{lihaoyi22@mails.tsinghua.edu.cn}

\begin{document}

\maketitle

\begin{abstract}

In this paper, we define and study a spin version of the open Toda integrable system, using the Iwasawa decomposition of a connected split real Lie group. We study spherical Whittaker functions as solutions of the spin Toda system. We also construct their N-point version and find the corresponding factorisation of the Casimir element, which yields asymptotic boundary KZB type equations.

\end{abstract}

\tableofcontents

\section{Introduction}

Representation theory of Lie groups plays a central role in the modern theory of integrable systems. Many classical integrable systems arise from invariant functions on Lie algebras or Lie groups, while their quantisations are obtained from invariant differential operators. Therefore, there is a deep connection between classical (resp. quantum) integrable systems and Poisson geometry (resp. harmonic analysis) on Lie groups. Familiar examples include Calogero-Moser systems, Toda systems, Ruijsenaars systems.

The Toda system is closely related (indeed, essentially equivalent) to the theory of Whittaker functions, as observed by Kostant\cite{kow,kos}. This representation-theoretic realisation provides the joint eigenfunctions of the commuting Toda Hamiltonians and extends the type A construction to arbitrary simple root systems and to various quantum-group analogues.

On the other hand, the Calogero-Moser system can be realized through symmetric spaces and spherical functions\cite{kaz}. In their ordinary form, spherical functions are complex-valued functions on a real Lie group $G$, which are bi-invariant with respect to the action of the maximal compact subgroup $K$. These systems serve as two of the most important exactly solvable quantum many-body problems. In \cite{res, sna}, a spin generalization of Calogero-Moser system was proposed by allowing arbitrary representations of the compact subgroup. (In the spirit of the orbit method, this amounts classically to allowing general orbits in the dual Lie algebra $\mathfrak{k}^*$).

In this paper, we use spherical Whittaker functions (instead of Whittaker functions) as solutions of Toda systems, and give a spin generalization of the usual open Toda chain. Let $G$ be a split real Lie group. Let $N$ denotes the unipotent subgroup of $G$, $A$ denotes the maximal split torus, $G=KAN$ be the Iwasawa decomposition. Suppose $(\sigma,V)$ is a finite dimensional representation of $K$ and $\chi$ is a character of $N$. Then the space of spherical Whittaker functions is defined as

$$
\begin{aligned}
&C^\infty((K,\sigma)\backslash G/(N,\chi))\\
=&\{f\in C^{\infty}(G)\otimes V: f(kan)=\frac{\sigma(k)}{\chi(n)}f(a),\ \forall k\in K, a\in A, n\in N\}
\end{aligned}
$$

It's easy to see that the restriction to the maximal split torus gives an isomorphism
$$
C^\infty((K,\sigma)\backslash G/(N,\chi))\cong C^\infty(A)\otimes V,\ \ \ f\mapsto \left.f\right|_A
$$
\begin{thm}
The twisted action of $\Omega\in Z(U(\mathfrak{g}))\cong\text{Diff}(G)^{G\times G}$ on $C^\infty(A)$ will give the spin Toda Hamiltonian
$$
H=-\frac{1}{2}\Delta-\sum_{i=1}^r\chi_i^2e^{-2\alpha_i(h)}-\sum_{i=1}^r\chi_i\sigma(k_i)e^{-\alpha_i(h)}
$$
where $\alpha_1,...,\alpha_r$ are the simple roots, $\chi_i=d\chi(e_{\alpha_i})$, and
$k_{i}=e_{\alpha_i}-e_{-\alpha_i}\in\mathfrak k$.
\end{thm}

When $\sigma$ is trivial, this reduces to the ordinary open quantum Toda lattice. This provides a uniform framework of many existing constructions of various Toda type systems. An example recovering $BC$-Toda system is explained in section 3. 

\begin{rem}
By construction, the wave functions of our spin Toda system are spherical-Whittaker functions. The analytic behaviour of such functions has been studied by Harish-Chandra\cite{har} in great detail. For more recent work, see van den Ban\cite{ban}, where, in particular, a radial part formula for general real reductive Lie groups is given.
\end{rem}

We also introduce an N-point version of spherical Whittaker function. Following the vertex-operator construction for N-point spherical functions, we compose differential intertwiners between principal series representations to construct the analogue for spherical Whittaker functions.  An appropriate factorisation of the quadratic Casimir element then produces a system of first-order asymptotic
boundary KZB type equations.
\begin{thm}
Let $U_1,...,U_N$ be finite dimensional $G$-modules and let $\mathcal{H}_{\lambda_0},...,\mathcal{H}_{\lambda_N}$ be principal series representations. Given differential vertex operators $D_i: \mathcal{H}_{\lambda_i}^\infty\to\mathcal{H}_{\lambda_{i-1}}^\infty\otimes U_i,\ 1\leq i\leq N$ and continuous intertwiners $\phi_l\in\text{Hom}_K(\mathcal{H}_{\lambda_0},V),\ \phi_r\in\text{Hom}_{N}(\mathbbm{C}_\chi,\mathcal{H}_{\lambda_N}^{-\infty})$, the restriction of the N-point spherical-Whittaker function to the maximal torus
$$
f(g):=(\phi_l\otimes id_{U_1\otimes...\otimes U_N})(D_1\otimes id_{U_2\otimes...\otimes U_N})...(D_{N-1}\otimes id_{U_N})D_N(\pi_{\lambda_N}(g)\phi_r)
$$
satisfies first-order equations
$$
\tilde{D}_i \left.f\right|_A=\left(\frac{(\lambda_i,\lambda_i)}{2}-\frac{(\lambda_{i-1},\lambda_{i-1})}{2}\right)\left.f\right|_A
$$
\end{thm}
Here the boundary KZB type operator $\tilde{D}_i$ is given by
$$
\tilde{D}_i=\sum_{j=1}^r(x_j)_{U_i}\partial_{x_j}-\tau^l_{VU_i}-\sum_{j=1}^{i-1}\tau^l_{U_jU_i}-d_{U_i}-\sum_{j=i+1}^N\tau^r_{U_iU_j}-\tilde{\tau}_{U_i,\chi^*}
$$
where
$$
\tau^l=\sum_{\alpha\in R^+}(k_\alpha\otimes e_\alpha)
$$
$$
\tau^r=\sum_{j=1}^n(x_j\otimes x_j)+2\sum_{\alpha\in R^+}e_\alpha\otimes e_\alpha-\sum_{\alpha\in R^+}k_\alpha\otimes e_\alpha
$$
$$
\tilde{\tau}=\sum_{\alpha\in R^+}a^{-\alpha}(e_\alpha+f_\alpha)\otimes e_\alpha
$$
$$
d=\frac{1}{2}\sum_{j=1}^nx_j^2-t_\rho+\sum_{\alpha\in R^+}e_{\alpha}^2
$$

We have chosen $R^+$ a set of positive roots. $e_\alpha\in\mathfrak{g}_{0,\alpha}$ and $f_\alpha\in\mathfrak{g}_{0,-\alpha}$ are normalized root vectors. $k_\alpha=e_\alpha-f_\alpha$ form a basis of $\mathfrak{k}_{0}$, and $\{x_j\}$ is an orthonormal basis of $\mathfrak{h}_0$.

\begin{rem}
The term ``KZB-type operator'' comes from the spin Calogero-Moser setting, where these operators have a natural quantum field theoretic interpretation, as the equation satisfied by the limits of correlation functions of the boundary Wess-Zumino-Witten model. It would be interesting to determine whether the spin Toda system arises this way in two- dimensional conformal field theory.
\end{rem}

Physically, N-point spherical Whittaker functions are the wave functions of spin Toda chains. In this paper we also construct the classical mechanical counterpart of this system and establish its superintegrability. This approach parallels to the construction of open and closed spin Calogero-Moser chains.

In section 2 we fix the Lie-theoretic conventions used throughout the paper. In section 3 we compute the radial part of the Casimir element acting on spherical Whittaker functions and obtain the spin Toda Hamiltonian. In section 4 we formalize the Iwasawa radial compoent computation, construct N-point spherical Whittaker functions and prove that they satisfy asymptotic boundary KZB type equations. Section 5 is concerned with the classical counterpart.

\section*{Acknowledgements}

The author is grateful to Professor Nicolai Reshetikhin for his kind support and for introducing him to the fields of mathematical physics and integrable systems. He also thanks Professor Jasper Stokman for pointing out the work of Harish-Chandra and for many other valuable suggestions and comments.

\section{Notations}
Throughout this paper all vector spaces, tensor products, and representations are understood over $\mathbbm{C}$, unless explicitly stated otherwise. Real Lie algebras are denoted by a subscript zero, and the complexification of a real Lie algebra $\mathfrak l_0$ is denoted by
$$
\mathfrak l=\mathfrak l_0\otimes_{\mathbbm{R}}\mathbbm{C}.
$$

Let $G$ be a connected real Lie group with finite centre whose Lie algebra $\mathfrak g_0$ is split and simple. Fix a Cartan involution $\theta_0\in\operatorname{Aut}(\mathfrak g_0)$ and write
$$
\mathfrak g_0=\mathfrak k_0\oplus\mathfrak p_0
$$
for the associated Cartan decomposition into the $+1$ and $-1$ eigenspaces of $\theta_0$, respectively. Let
$$
\mathfrak g=\mathfrak g_0\otimes_{\mathbbm{R}}\mathbbm{C},\qquad\mathfrak k=\mathfrak k_0\otimes_{\mathbbm{R}}\mathbbm{C},\qquad\mathfrak p=\mathfrak p_0\otimes_{\mathbbm{R}}\mathbbm{C},
$$
and denote the complex-linear extension of $\theta_0$ to $\mathfrak g$ by $\theta$. Thus $\mathfrak g=\mathfrak k\oplus\mathfrak p$. Let $K\subseteq G$ be the maximal compact subgroup with Lie algebra $\mathfrak k_0$.

Fix a maximal abelian subspace $\mathfrak h_0\subseteq\mathfrak p_0$. Since $\mathfrak g_0$ is split, $\mathfrak h_0$ is a split Cartan subalgebra of $\mathfrak g_0$, and its complexification
$$
\mathfrak h=\mathfrak h_0\otimes_{\mathbbm{R}}\mathbbm{C}
$$
is a Cartan subalgebra of $\mathfrak g$. Let
\begin{equation}\label{eq:root-decomposition}
\mathfrak g=\mathfrak h\oplus\bigoplus_{\alpha\in R}\mathfrak g_\alpha,\qquad\mathfrak g_\alpha=\{x\in\mathfrak g:[h,x]=\alpha(h)x\text{ for all }h\in\mathfrak h\}.
\end{equation}
Here $R=R(\mathfrak g,\mathfrak h)\subseteq\mathfrak h^*$ is the root system. Since $\mathfrak g_0$ is split, it agrees with the restricted root system of $(\mathfrak g_0,\mathfrak h_0)$, and $\mathfrak g_{0,\alpha}:=\mathfrak g_0\cap\mathfrak g_\alpha$ is one-dimensional over $\mathbbm{R}$ for every $\alpha\in R$.

Fix a positive system $R_+\subseteq R$, and denote its simple roots by
$$
\alpha_1,\ldots,\alpha_r,\qquad r=\dim\mathfrak h_0.
$$
We write $Q=\mathbbm{Z} R$ for the root lattice and
$$
W=N_K(\mathfrak h_0)/Z_K(\mathfrak h_0)
$$
for the Weyl group.

Let $\langle\cdot,\cdot\rangle_{\mathfrak g_0}$ be the Killing form of $\mathfrak g_0$, and use the same notation for its complex-bilinear extension to $\mathfrak g$. Its restriction to $\mathfrak h_0$ is positive definite. For $\lambda\in\mathfrak h^*$, let $t_\lambda\in\mathfrak h$ be the unique element determined by
$$
\langle h,t_\lambda\rangle_{\mathfrak g}=\lambda(h),\qquad h\in\mathfrak h.
$$
We use the induced bilinear form on $\mathfrak h^*$,
$$
(\lambda,\mu):=\langle t_\lambda,t_\mu\rangle_{\mathfrak g}.
$$
In particular, $t_\alpha\in\mathfrak h_0$ for every $\alpha\in R$.

Choose root vectors $e_\alpha\in\mathfrak g_{0,\alpha}$, $\alpha\in R$, normalised by
$$
[e_\alpha,e_{-\alpha}]=t_\alpha,\qquad\theta_0(e_\alpha)=-e_{-\alpha}.
$$
Equivalently, $\langle e_\alpha,e_{-\alpha}\rangle_{\mathfrak g_0}=1$. For $\alpha\in R_+$, set
$$
f_\alpha:=e_{-\alpha},\qquad k_\alpha:=e_\alpha-f_\alpha.
$$
Then
$$
[e_\alpha,f_\alpha]=t_\alpha,\qquad k_\alpha\in\mathfrak k_0,
$$
and
$$
\mathfrak k_0=\bigoplus_{\alpha\in R_+}\mathbbm{R} k_\alpha,\qquad\mathfrak p_0=\mathfrak h_0\oplus\bigoplus_{\alpha\in R_+}\mathbbm{R}(e_\alpha+f_\alpha).
$$

Set
$$
\rho:=\frac12\sum_{\alpha\in R_+}\alpha\in\mathfrak h^*,\qquad t_\rho=\frac12\sum_{\alpha\in R_+}t_\alpha\in\mathfrak h.
$$
Thus $\rho$ is a linear functional, whereas $t_\rho$ is an element of $\mathfrak h$. We write $\rho(h)$, equivalently $\langle h,t_\rho\rangle_{\mathfrak g}$, for their pairing.

Let
$$
A:=\exp(\mathfrak h_0),\qquad\mathfrak n_{0,+}:=\bigoplus_{\alpha\in R_+}\mathfrak g_{0,\alpha},\qquad\mathfrak n:=\mathfrak n_{0,+}\otimes_{\mathbbm{R}}\mathbbm{C}.
$$
Let $N\subseteq G$ be the connected subgroup with Lie algebra $\mathfrak n_{0,+}$. The multiplication map
$$
K\times A\times N\longrightarrow G,\qquad(k,a,n)\longmapsto kan,
$$
is a diffeomorphism. This is the global Iwasawa decomposition. We also write
$$
M:=Z_K(\mathfrak h_0),\qquad P:=MAN
$$
for the centraliser of $\mathfrak h_0$ in $K$ and the corresponding minimal parabolic subgroup.

For $\lambda\in\mathfrak h^*$, define
$$
\xi_\lambda(a)=a^\lambda:=e^{\lambda(\log a)},\qquad a\in A.
$$
Thus, if $a=e^h$ with $h\in\mathfrak h_0$, then
$$
a^\lambda=e^{\lambda(h)}.
$$
Our adjoint-action convention is
$$
\Ad_a(e_\alpha)=a^\alpha e_\alpha,\qquad\Ad_{a^{-1}}(e_\alpha)=a^{-\alpha}e_\alpha.
$$

Choose an orthonormal basis $x_1,\ldots,x_r$ of $\mathfrak h_0$ with respect to the Killing form. For $x\in\mathfrak h_0$, define
$$
(\partial_xu)(a):=\left.\frac{d}{dt}\right|_{t=0}u(a\exp(tx)),\qquad u\in C^\infty(A).
$$
We write
$$
\Delta=\sum_{j=1}^r\partial_{x_j}^2,\qquad\partial_{t_\rho}=\sum_{j=1}^r\rho(x_j)\partial_{x_j}.
$$

We write $U(\mathfrak g)$ and $S(\mathfrak g)$ for the universal enveloping and symmetric algebras of $\mathfrak g$, respectively, and $Z(U(\mathfrak g))$ for the centre of $U(\mathfrak g)$. The antipode of $U(\mathfrak g)$ is denoted by $S$ and is determined by $S(x)=-x$ for $x\in\mathfrak g$.

The famous Harish-Chandra and Chevalley isomorphisms give
$$
\Diff(G)^{G\times G}\simeq Z(U(\mathfrak g))\simeq S(\mathfrak h)^W,\qquad S(\mathfrak g)^G\simeq S(\mathfrak h)^W.
$$

For $g\in G$, define the left and right regular actions on smooth
functions by
$$
(l(g)F)(g')=F(g^{-1}g')
$$
$$
(r(g)F)(g')=F(g'g).
$$
Their infinitesimal actions are denoted by $L$ and $R$:
$$
(L(x)F)(g)=\left.\frac{d}{dt}\right|_{t=0}F(\exp(-tx)g),
$$
$$
(R(x)F)(g)=\left.\frac{d}{dt}\right|_{t=0}F(g\exp(tx)).
$$
For any invariant differential operator $P\in Z(U(\mathfrak{g}))$, since it bi-invariant, we have
$$
L(P)=R(P).
$$

We also use the tensor-leg notation for tensor products, for example, if
$$
T=\sum_s a_s\otimes b_s\in\mathfrak{g}\otimes\mathfrak{g}
$$
and $X,Y$ are $\mathfrak{g}$-modules, we write
$$
T_{XY}:=\sum_s(a_s)_X(b_s)_Y
$$
for the operator acting on the corresponding factors.

\section{Construction of the spin Toda operator}

Recall that the theorems of Harish-Chandra and Chevalley give
$$
\mathcal{D}(G)^{G\times G}\cong Z(U(\mathfrak{g}))\cong S\mathfrak{g}^G\cong S\mathfrak{h}^W\cong \mathbbm{C}[x_1,...,x_r]
$$
which allows us to carry out explicit computations.

Here $r=rank(G)$ and we may choose $x_1$ to correspond to the Casimir element
$$
\Omega=\sum x_i^2+2\sum_{\alpha>0}f_\alpha e_\alpha+2t_\rho
$$

The action of $\Omega$ on spherical functions
$$
C^\infty(K\backslash G/ K)\subset C^\infty(G)
$$
gives the quantum Calogero-Moser system.

If we instead consider
$$
C^\infty((N_+,\chi_+)\backslash G/(N_-,\chi_-))\subset C^\infty (G)
$$
then the action of the Casimir element gives rise to the quantum Toda system
$$
L(\Omega)=R(\Omega)=\sum\partial_i^2-2\partial_\rho+2\sum_{i=1}^r\chi_+^i\chi_-^i e^{\alpha_i(h)}
$$

We now consider a mixture of the two cpnstructions above. Namely, we consider the space
$$
C^\infty((K,\sigma)\backslash G/(N,\chi))\subset C^\infty(G)\otimes V
$$
here $(\sigma,V)$ is a finite dimensional representation of $K$. The action of $Z(U(\mathfrak{g}))$ on it will still give a superintegrable system, which we now compute explicitly.

Namely,
$$
\begin{aligned}
&C^\infty((K,\sigma)\backslash G/(N,\chi))\\
=&\{f\in C^{\infty}(G)\otimes V: f(kan)=\frac{\sigma(k)}{\chi(n)}f(a),\ \forall k\in K, a\in A, n\in N\}
\end{aligned}
$$

The Lie algebra $\mathfrak{k}=\text{Lie}(K)$ is generated by $k_\alpha=e_{\alpha}-f_{\alpha}$, and therefore
$$
\Omega=\sum x_i^2+2\sum_{\alpha>0}e_\alpha^2-2\sum_{\alpha>0}k_\alpha e_\alpha+2t_\rho
$$

For $\phi\in C^\infty((K,\sigma)\backslash G/(N,\chi))$, by definition we have:
$$
\begin{aligned}
(L(k_\alpha)\phi)(e^h)&=\left.\frac{\text{d}}{\text{d}t}\right|_{t=0}\phi(\text{exp}(-tk_{\alpha})e^h)\\
&=-\sigma(k_{\alpha})\phi(e^h)
\end{aligned}
$$
$$
\begin{aligned}
(L(e_\alpha)\phi)(e^h)&=\left.\frac{\text{d}}{\text{d}t}\right|_{t=0}\phi(\text{exp}(-te_{\alpha})e^h)\\
&=\left.\frac{\text{d}}{\text{d}t}\right|_{t=0}\phi(e^h\cdot e^{-h}\text{exp}(-te_{\alpha})e^h)\\
&=\left.\frac{\text{d}}{\text{d}t}\right|_{t=0}\phi(e^h\cdot \text{exp}(-te^{-\alpha(h)}e_{\alpha}))\\
&=e^{-\alpha(h)}\chi(e_{\alpha})\phi(e^h)
\end{aligned}
$$

Similarly, we can explicitly compute the action of $R(k_\alpha)$ and $R(e_\alpha)$, hence
$$
L(\Omega)=R(\Omega)=\sum_{i=1}^r\partial_i^2+2\partial_\rho+2\sum_{i=1}^r\chi_i^2e^{-2\alpha_i(h)}+2\sum_{i=1}^r\chi_i\sigma(k_i)e^{-\alpha_i(h)}
$$
here $\chi_i=\chi(e_{\alpha_i})$. This expression only consists of simple roots because $\chi$ is a character and
$$
\mathfrak{n}/[\mathfrak{n},\mathfrak{n}]\cong\bigoplus_{i=1}^r\mathbbm{C}e_{\alpha_i}
$$
Conjugating by $e^{\rho(h)}$, we get
$$
\begin{aligned}
D&=e^{\rho(h)}L(\Omega)e^{-\rho(h)}\\
&=\sum_{i=1}^r\partial_i^2+2\sum_{i=1}^r\chi_i^2e^{-2\alpha_i(h)}+2\sum_{i=1}^r\chi_i\sigma(k_i)e^{-\alpha_i(h)}-(\rho,\rho)
\end{aligned}
$$

Therefore the operator
$$
\begin{aligned}
H&=-\frac{1}{2}(D+(\rho,\rho))\\
&=-\frac{1}{2}\sum_{i=1}^r\partial_i^2-\sum_{i=1}^r\chi_i^2e^{-2\alpha_i(h)}-\sum_{i=1}^r\chi_i\sigma(k_i)e^{-\alpha_i(h)}
\end{aligned}
$$
is a standard second order differential operator, with spectrum bounded below if $\chi$ is purely imaginary and $\sigma$ is unitary, hence it defines a quantum Hamiltonian. If the representation $\sigma$ is trivial, we recover the usual quantum Toda lattice. Therefore the Hamiltonian system we get is a spin version of the Toda lattice.

\begin{exmp}
We first consider the simplest non-trivial example, in which $\sigma$ is a non-trivial one-dimensional representation. By \cite{sgo}, this happens only when $\mathfrak{g}$ is of type C.

Let $\mathfrak{g}=\mathfrak{sp}_{2r}$. Then
$$
\mathfrak{h}=\bigoplus_{i=1}^{r}\mathbbm{C}(E_{i,i}-E_{i+r,i+r})
$$

The root system is
$$
R=\{\pm(\epsilon_k\pm\epsilon_l)\}_{k\neq l}\cup\{\pm 2\epsilon_j\}
$$
with simple roots
$$
\alpha_k=\epsilon_k-\epsilon_{k+1},\ \ 1\leq k<r,\ \ \alpha_r=2\epsilon_r
$$
here
$$
\epsilon_{j}(E_{i,i}-E_{i+r,i+r})=\delta_{ij}
$$

Recall that $\mathfrak{k}\cong\mathfrak{gl}_r$ and that the space of characters $\text{ch}(\mathfrak{k})$ is one-dimensional, generated by the character $\chi_{\mathfrak{sp}}$ determined by
$$
\chi_{\mathfrak{sp}}(k_i)=0,\ 1\leq i<r,\ \ \  \chi_{\mathfrak{sp}}(k_r)=1
$$
therefore, set $\sigma=d\chi_{\mathfrak{sp}}$ and
$$
\chi_i=\sqrt{-1},\ 1\leq i<r,\ \ \chi_r=c
$$
we get
$$
H=-\frac{1}{2}\Delta+\sum_{i=1}^{r-1}e^{2(x_i-x_{i+1})}-c^2e^{4x_r}+cde^{2x_r}
$$

Collecting the coefficients and performing the coordinate transformation $x\mapsto 2x$, we obtain
$$
H=-\frac{1}{2}\Delta+e^{x_1-x_2}+...+e^{x_{r-1}-x_r}+Ce^{x_r}+e^{2x_r}
$$

Both $c$ and $d$ should be purely imaginary. This is an integrable system, which can also be constructed by shifting the Lax matrix of the usual open Toda chain by $C\chi_{\mathfrak{sp}}\in\mathfrak{k}^*$\cite{rey}.
\end{exmp}

\section{Factorisation of the Casimir and boundary KZB type operators}

In this section, we first carry out the Iwasawa radial part computation, and then give a factorisation of the Casimir element which is compatible with spherical Whittaker functions. This will give an explicit form of the first order differential equations.

\subsection{The Iwasawa radial-component map}
\ 

Let
$$
R:=\text{span}\{\xi_{\lambda}:\lambda\in Q\}\cong\mathbbm{C}[Q]
$$
and consider the space
$$
V=U(\mathfrak{h})\otimes U(\mathfrak{k})\otimes U(\mathfrak{n})
$$
For any $a\in A$, we define a map
$$
\tilde{\Gamma}_a: R\otimes V\to U(\mathfrak{g})\ \ \ 
f\otimes h\otimes x\otimes y\mapsto f(a)\text{Ad}_{a^{-1}}(x)hy
$$

\begin{thm}
There exists a linear map (the Iwasawa radial component)
$$
\Pi: U(\mathfrak{g})\to R\otimes V
$$
such that
$$
\tilde{\Gamma}_a\circ\Pi=id_{U(\mathfrak{g})},\ \ \ \forall a\in A
$$
\end{thm}

\begin{proof}
This theorem is essentially a reformulation of the infinitesimal Iwasawa decomposition: for every $a\in A$, we have
$$
\mathfrak{g}=\text{Ad}_{a^{-1}}(\mathfrak{k})\oplus\mathfrak{h}\oplus\mathfrak{n}
$$
In particular, every monomial $f_\alpha\in\mathfrak{n}_{-}$ can be written in terms of $\text{Ad}_{a^{-1}}(k_\alpha), e_\alpha$ and elements in $\mathfrak{h}$ (with coefficients possibly rational functions in $a^\alpha$). By the PBW theorem
$$
U(\mathfrak{g})=U(\mathfrak{n}_-)U(\mathfrak{h})U(\mathfrak{n})
$$
so any element $u\in U(\mathfrak{g})$ can be written uniquely in this form. We express $f_\alpha$ in terms of the infinitesimal decomposition and use the rule
$$
[h,e_\alpha]=\alpha(h)e_\alpha,\ \ \forall h\in\mathfrak{h}
$$
to move all $e_\alpha$ to the right of elements of $\mathfrak{h}$. Replacing each $a^\lambda$ by the function $\xi_\lambda$, we get a map
$$
\Pi: U(\mathfrak{g})\to\mathbbm{C}(Q)\otimes V
$$

Notice that $\text{Ad}_{a^{-1}}(k_\alpha)=a^{-\alpha}e_\alpha-a^\alpha e_{-\alpha}$, so
$$
e_{-\alpha}=a^{-2\alpha}e_\alpha-a^{-\alpha}\text{Ad}_{a^{-1}}(k_\alpha)
$$
involves only coefficients from the polynomial ring $\mathbbm{C}[Q]\subset\mathbbm{C}(Q)$. This implies that $\Pi$ takes values in $R\otimes V$, and it follows directly from the construction that
$$
\tilde{\Gamma}_a\circ\Pi=\text{id}_{U(\mathfrak{g})}
$$
\end{proof}

\begin{rem}
This proof actually extends to more general situations when we enlarge the coefficient ring $R$. For example, a similar constuction for symmetric spaces is performed in \cite{sna} to derive the spin Calogero-Moser Hamiltonians.
\end{rem}

Now consider the Casimir element
$$
\Omega=\sum_{j=1}^nx_j^2+2t_\rho+2\sum_{\alpha\in R^+}e_{-\alpha}e_{\alpha}
$$

Substituting $e_{-\alpha}=a^{-2\alpha}e_\alpha-a^{-\alpha}\text{Ad}_{a^{-1}}(k_\alpha)$ into this expression, we get
$$
\Omega=\sum_{j=1}^nx_j^2+2t_\rho+2\sum_{\alpha\in R^+}a^{-2\alpha}e_\alpha^2-2\sum_{\alpha\in R^+}a^{-\alpha}\text{Ad}_{a^{-1}}(k_\alpha)e_{\alpha}
$$
hence
$$
\begin{aligned}
\Pi(\Omega)&=\sum_{j=1}^n1\otimes x_j^2\otimes1\otimes1+2\cdot1\otimes t_\rho\otimes 1\otimes 1\\
&+2\sum_{\alpha\in R^+}\left(\xi_{-2\alpha}\otimes 1\otimes1\otimes e_\alpha^2-\xi_{-\alpha}\otimes 1\otimes k_\alpha\otimes e_\alpha\right)
\end{aligned}
$$

Note that this actually formalizes the computation of the previous section.

\subsection{Factorisation and vertex operators}

Let $E$, $E'$ and $U$ be $\mathfrak{g}$-modules. By definition, a vertex operator from $E$ to $E'\otimes U$ is an intertwiner
$$
\Phi\in\text{Hom}_{\mathfrak{g}}(E,E'\otimes U)
$$

The vector space $\text{Hom}_{\mathbbm{C}}(E,E'\otimes U)$ carries the following two natural actions of $U(\mathfrak{g})^{\otimes 2}$:
$$
\left((x\otimes y)\Phi\right)(v)=(x_{E'}\otimes y_U)\Phi(v)
$$
$$
\left(\Phi*(x\otimes y)\right)(v)=(1_{E'}\otimes S(x)_U)\Phi(y_E v)
$$
\begin{defn}
A triple $(\tau^l,\tau^r,d)$ with $\tau^l,\tau^r\in\mathfrak{g}\otimes\mathfrak{g}$ and $d\in U(\mathfrak{g})$ gives a factorisation of $\Omega$ if for all $\mathfrak{g}$-modules $E,E',U$ and intertwiners $\Phi\in\text{Hom}_{\mathfrak{g}}(E,E'\otimes U)$ we have
$$
\frac{1}{2}((\Omega\otimes 1)\Phi-\Phi\Omega)=\tau^l\Phi+\Phi*\tau^r+(1\otimes d)\Phi
$$
here $\Omega$ is regarded as an element of $Z(U(\mathfrak{g}))\subset U(\mathfrak{g})$
\end{defn}

\begin{rem}
For simplicity, our convention differs from that of \cite{sna}: our right tensor is $\tau^r=-\tau_{Reshetikhin-Stokman}^r$. This convention is used throughout the rest of the paper.
\end{rem}

We now recall the definitions of several special kind of representations, that enter in the definition of N-point spherical Whittaker functions.
\begin{defn}
Let $\mathfrak{b}=\mathfrak{h}\oplus\mathfrak{n}$ be the Borel subalgebra. For any weight $\mu$, let $\mathbbm{C}_\mu$ be the one-dimensional $\mathfrak{b}$-module on which $\mathfrak{h}$ acts through $\mu$ and $\mathfrak{n}$ acts trivially. Then
$$
M_\mu:=U(\mathfrak{g})\otimes_{U(\mathfrak{b})}\mathbbm{C}_\mu\in\mathfrak{g}-\text{mod}
$$
is the Verma module of highest weight $\mu$
\end{defn}

\begin{defn}
For $\lambda\in\mathfrak{h}^*$, consider the pre-Hilbert space
$$
U_\lambda:=\{F\in C^0(G):F(gman)=a^{-\lambda-\rho}F(g),\ \forall m\in M, a\in A, n\in N\}
$$
with scalar product
$$
\langle F, G\rangle_{\lambda}:=\int_K F(g)\overline{G(g)}\text{d}g
$$
and $G$-action
$$
(\pi_\lambda(g_0)F)(g):=F(g_0^{-1}g)
$$
Then the Hilbert space completion $\mathcal{H}_\lambda:=\widehat{U_\lambda}$ is called the principal series representation.
\end{defn}

Differential vertex operators for principal series representations are closely related to vertex operators for Verma modules, by the following proposition in \cite{sna}:
\begin{prop}
For any finite dimensional $G$-module $U$, there is a natural isomorphism:
$$
\mathcal{D}_G(\mathcal{H}_{\lambda}^\infty,\mathcal{H}_{\mu}^\infty\otimes U)\cong\text{Hom}_{\mathfrak{g}}(M_{-\mu-\rho},M_{-\lambda-\rho}\otimes U)
$$
\end{prop}

Also, notice that with our normalization of the Casimir element:
$$
\Omega|_{\mathcal{H}_\lambda^\infty}=\left((\lambda,\lambda)-(\rho,\rho)\right)\text{id}_{\mathcal{H}_\lambda^\infty}
$$

We are now ready to define N-point spherical Whittaker functions.

\begin{defn}
Let $U_1,...,U_N$ be finite dimensional $G$-representations, $\lambda_0, and let \lambda_1,...,\lambda_N$ be weights. Suppose we are given differential vertex operators
$$
D_i\in\mathcal{D}_G(\mathcal{H}_{\lambda_i}^\infty,\mathcal{H}_{\lambda_{i-1}}^\infty\otimes U_i)\cong\text{Hom}_{\mathfrak{g}}(M_{-\lambda_{i-1}-\rho},M_{-\lambda_i-\rho}\otimes U_i)
$$

Let $(\sigma,V)$ be a finite dimensional representation of $K$ and $\chi:N\to\mathbbm{C}^\times$ a character, suppose $\phi_l\in\text{Hom}_K(\mathcal{H}_{\lambda_0},V)$ and $\phi_r\in\text{Hom}_{N}(\mathbbm{C}_\chi,\mathcal{H}_{\lambda_N}^{-\infty})$. The N-point spherical-Whittaker function is defined as
$$
f(g):=(\phi_l\otimes id_{U_1\otimes...\otimes U_N})(D_1\otimes id_{U_2\otimes...\otimes U_N})...(D_{N-1}\otimes id_{U_N})D_N(\pi_{\lambda_N}(g)\phi_r),\ \ g\in G
$$
\end{defn}

It's easy to see via the intertwining properties that $e^\rho f(g)$ is actually a wave function of the spin Toda system associated with $(\sigma^{(N)},\chi^*)$ with energy $-\frac{1}{2}(\lambda_N,\lambda_N)$, here the N-fold tensor product representation is
$$
\sigma^{(N)}:=\sigma\otimes\pi_{U_1}\otimes\pi_{U_2}\otimes...\otimes\pi_{U_{N}}
$$

\begin{rem}
Notice that, for simplicity, by our definition the N-point spherical-Whittaker function is actually associated to $\sigma^(N)$ and $\chi^*=\chi^{-1}$, instead of $\chi$.
\end{rem}

We now prove the main theorem.

\begin{thm}
Suppose that
$$
a\mapsto(\tau^l(a),\tau^r(a),d(a))
$$
is a family of triples such that, for every $a\in A$, the triple gives a factorisation of $\Omega$.

Assume that
$$
\tau^l\in\mathfrak{k}\otimes\mathfrak{g}\ \ \ \ \tilde{\tau}\in\mathfrak{g}\otimes\mathfrak{n}
$$
Here $\tilde{\tau}$ is defined by
$$
(1\otimes\text{Ad}_{a^{-1}})\tau^r=\sum_{j=1}^nx_j\otimes x_j+\tilde{\tau}(a)
$$

Then the restriction of $f$ to $A$ satisfies the asymptotic boundary KZB type equations
$$
\tilde{D}_i f=\left(\frac{(\lambda_i,\lambda_i)}{2}-\frac{(\lambda_{i-1},\lambda_{i-1})}{2}\right)f
$$
here the boundary KZB type operator $\tilde{D}_i$ is given by
$$
\tilde{D}_i=\sum_{j=1}^n(x_j)_{U_i}\partial_{x_j}-\tau^l_{VU_i}-\sum_{j=1}^{i-1}\tau^l_{U_jU_i}-d_{U_i}-\sum_{j=i+1}^N\tau^r_{U_iU_j}-\tilde{\tau}_{U_i,\chi^*}
$$
\end{thm}

\begin{proof}
The proof follows directly from \cite{sfs,sna}, here we only sketch it. Applying the left action and the right $*$-action of $(\Omega\otimes 1)$ to the $i$-th factor of the N-point function, we get the eigenvalue
$$
\frac{1}{2}\left(((\lambda_{i-1},\lambda_{i-1})-(\rho,\rho))-((\lambda_{i},\lambda_{i})-(\rho,\rho))\right)=\frac{1}{2}\left((\lambda_{i-1},\lambda_{i-1})-(\lambda_i,\lambda_i)\right)
$$.
Using the intertwining properties, we can move the action of $\tau^l$ to $U_1\otimes...\otimes U_{i-1}$ and to the spherical boundary $\sigma$, and the action of $\tau^r$ to $U_{i+1}\otimes...\otimes  U_N$.

To reach the Whittaker boundary $\mathbbm{C}_\chi$, recall that
$$
X\pi_{\lambda_N}(a)=\pi_{\lambda_N}(a)\text{Ad}_{a^{-1}}(X),\ \ \forall X\in U(\mathfrak{g})
$$
Therefore the action of $\tau^r$ is intertwined into the action of $\tilde{\tau}$ on $\mathbbm{C}_\chi$, while the remaining Cartan term gives the first- order differential operator
$$
E=\sum_{j=1}^n\partial_{x_j}\otimes x_j
$$

The theorem is proved by reordering the terms.
\end{proof}

The key result of this section is an explicit factorisation.

Since $\Omega$ is central, we have:
$$
(\Omega\otimes 1)\Phi-\Phi\Omega=(\text{Ad}_a(\Omega)\otimes 1)\Phi-\Phi\text{Ad}_a(\Omega),\ \ \forall a\in A
$$
Substituting the radial part computation into the right-hand side:

$$
\begin{aligned}
(\Omega\otimes 1)\Phi-\Phi\Omega
=&\sum_{j=1}^n\left((x_j^2\otimes 1)\Phi-\Phi x_j^2\right)+2\left((t_\rho\otimes 1)\Phi-\Phi t_\rho\right)\\
&+2\sum_{\alpha\in R^+}a^{-2\alpha}\left((\text{Ad}_a(e_\alpha^2)\otimes 1)\Phi-\Phi\text{Ad}_a(e_{\alpha}^2)\right)\\
&-2\sum_{\alpha\in R^+}a^{-\alpha}\left((k_\alpha\text{Ad}_a(e_\alpha)\otimes 1)\Phi-\Phi k_\alpha\text{Ad}_a(e_\alpha)\right)
\end{aligned}
$$

Using the intertwining property, we get
$$
\begin{aligned}
(x_j^2\otimes 1)\Phi-\Phi x_j^2=&-2(1\otimes x_j)\Phi x_j+(1\otimes x_j^2)\Phi
\\
(t_\rho\otimes 1)\Phi-\Phi t_\rho=&-(1\otimes t_\rho)\Phi
\\
(\text{Ad}_a(e_\alpha^2)\otimes 1)\Phi-\Phi\text{Ad}_a(e_{\alpha}^2)
=&-2(1\otimes\text{Ad}_a(e_{\alpha}))\Phi\text{Ad}_a(e_{\alpha})+(1\otimes \text{Ad}_a(e_{\alpha}^2))\Phi\\
=&-2a^{2\alpha}(1\otimes e_\alpha)\Phi e_\alpha+a^{2\alpha}(1\otimes e_\alpha^2)\Phi
\\
(k_\alpha\text{Ad}_a(e_\alpha)\otimes 1)\Phi-\Phi k_\alpha\text{Ad}_a(e_\alpha)
=&-(1\otimes k_\alpha)\Phi\text{Ad}_a(e_{\alpha})-(k_\alpha\otimes\text{Ad}_a(e_{\alpha}))\Phi\\
=&-a^\alpha(1\otimes k_\alpha)\Phi e_\alpha-a^\alpha(k_\alpha\otimes e_\alpha)\Phi
\end{aligned}
$$

Therefore we obtain the following
\begin{prop}
The following triple is a factorisation of $\Omega$:
$$
\tau^l=\sum_{\alpha\in R^+}(k_\alpha\otimes e_\alpha)
$$
$$
\tau^r=\sum_{j=1}^n(x_j\otimes x_j)+2\sum_{\alpha\in R^+}e_\alpha\otimes e_\alpha-\sum_{\alpha\in R^+}k_\alpha\otimes e_\alpha
$$
$$
d=\frac{1}{2}\sum_{j=1}^nx_j^2-t_\rho+\sum_{\alpha\in R^+}e_{\alpha}^2
$$
\end{prop}

This factorisation satisfies the hypotheses and hence gives the boundary KZB-type equation for N-point spherical Whittaker functions.

Since $k_\alpha=e_\alpha-f_\alpha$, the right tensor
may equivalently be written as
$$
\tau^r=\sum_{j=1}^r x_j\otimes x_j+\sum_{\alpha\in R_+}(e_\alpha+f_\alpha)\otimes e_\alpha
$$
Consequently, for this constant factorization we have
$$
\tilde\tau(a)=\sum_{\alpha\in R_+}a^{-\alpha}(e_\alpha+f_\alpha)\otimes e_\alpha
$$

Substituting the explicit formula into the equation gives the fully expanded operator
$$
\begin{aligned}
\tilde D_i=&\sum_{s=1}^r(x_s)_{U_i}\partial_{x_s}-\sum_{\alpha\in R_+}(k_\alpha)_V(e_\alpha)_{U_i}-\sum_{j=1}^{i-1}\sum_{\alpha\in R_+}(k_\alpha)_{U_j}(e_\alpha)_{U_i}\\
  &-\left(
    \frac12\sum_{s=1}^r x_s^2-t_\rho
    +\sum_{\alpha\in R_+}e_\alpha^2
    \right)_{U_i}\\
  &-\sum_{j=i+1}^N\left(
    \sum_{s=1}^r(x_s)_{U_i}(x_s)_{U_j}
    +\sum_{\alpha\in R_+}
     (e_\alpha+f_\alpha)_{U_i}(e_\alpha)_{U_j}
    \right)\\
  &+\sum_{p=1}^r
    \chi_p a^{-\alpha_p}
    (e_{\alpha_p}+f_{\alpha_p})_{U_i}.
\end{aligned}
$$

\section{Classical superintegrable spin Toda chains}

In this section we construct the classical counterpart of our construction.

Recall that a classical integrable system on a symplectic manifold $(M,\omega)$ of dimension $2n$ may be regarded as a Lagrangian fibration:
$$
p: M\to B
$$
here $p$ is a Poisson map to $B$ which is equipped with the trivial Poisson structure, and the fibers are Largagian submanifolds of $M$. It follows that $\text{dim} B=\frac{1}{2}\text{dim}M$. Algebraically, this amounts to finding a Poisson-commutative subalgebra
$$
I:=p^*C^\infty(B)\subset C^\infty(M)
$$
with $n$ independent generators $\{I_1,...,I_n\}$. Typically, $I_1$ is quadratic in the momenta and serves as the Hamiltonian.

Analogously, a superintegrable system consists of the following data:
$$
p_1: M\to P,\ \ \ \ p_2: P\to B
$$
here $p_1$ and $p_2$ are surjective Poisson maps, and $B$ is again equipped with the trivial Poisson structure. The superintegrability condition is
$$
\text{dim}M=\text{dim}P+\text{dim}B
$$
A generic fibre of $p_2$ is then a disjoint union of symplectic leaves of $P$. Algebraically, this corresponds to:
$$
I=p_1^*p_2^* C^\infty(B)\subset J=p_1^* C^\infty(P)\subset A=C^\infty(M)
$$
here $I$ is the algebra of commuting Hamiltonians, with $k\leq n$ independent generators $I_1,...,I_k$, $J$ is the algebra of integrals of motion, with $2n-k$ independent generators $J_1,...,J_{2n-k}$, subject to the conditions
$$
\{I_r,I_s\}=0,\ \ \ \{I_r,J_t\}=0,\ \ \forall 1\leq r,s\leq k, 1\leq t\leq 2n-k
$$

Notice that when $k=n$, the algebras $I$ and $J$, and hence $B=\text{Spec}I$ and $P=\text{Spec}J$ coincide, and we recover ordinary Liouville integrable systems.

In general, a superintegrable system is still exactly solvable, with invariant tori of even lower dimension, such systems theresore play an important role in mathematical physics. A large family of superintegrable systems, including spin Calogero-Moser chains and spin Ruijsenaars models, is constructed from symplectic leaves of moduli spaces of flat connections on Riemann surfaces. For a comprehensive account of these constructions, see \cite{rsc}.

We now apply this method to construct the classical spin Toda chain.

The phase spaces are symplectic leaves of $T^*(G^{n+1})/(K\times G^n\times N)$, here the action of $\mathbf{G}:=K\times G^n\times N$ is induced by the following action on $G^{n+1}$:
$$
(k,h_1,...,h_n,n)\cdot(g_0,...,g_n)=(kg_0h_1^{-1},h_1g_1h_2^{-1},...,h_ng_nn^{-1})
$$
After trivializing $T^*(G^{n+1})\cong\mathfrak{g}^{*n+1}\times G^{n+1}$ by right translation, the action can be written as
$$
\begin{aligned}
&(k,h_1,...,h_n,n)\cdot(x_0,...,x_n,g_0,...,g_n)\\
=&(\text{Ad}_k^*(x_0),\text{Ad}_{h_1}^*(x_1),...,\text{Ad}_{h_n}^*(x_n),kg_0h_1^{-1},h_1g_1h_2^{-1},...,h_ng_nn^{-1})
\end{aligned}
$$

This action is Hamiltonian with respect to the canonical symplectic structure on the cotangent bundles, with moment map given by
$$
\begin{aligned}
\mu:T^*(G^{n+1})&\to \mathfrak{k}^*\times \mathfrak{g}^{*n}\times\mathfrak{n}^*\\
(x_0,x_1,...,x_n,g_0,g_1,...,g_n)&\mapsto (\pi_{\mathfrak{k}}(x_0),x_1-\text{Ad}_{g_0^{-1}}^*(x_0),...,x_n-\text{Ad}_{g_{n-1}^{-1}}^*(x_{n-1}),-\pi_{\mathfrak{n}}(\text{Ad}_{g_n^{-1}}^*x_n))
\end{aligned}
$$
Here for any subalgebra $\mathfrak{l}\subset\mathfrak{g}$ we denote by $\pi_{\mathfrak{l}}: \mathfrak{g}^*\to\mathfrak{l}^*$ the corresponding projection.

This action is clearly the combination of the left action of $K\times G^n$ and the right action of $G^n\times N$, which are both Hamiltonian and commute with each other. The moment maps are
$$
\begin{aligned}
\mu_L:T^*(G^{n+1})&\to \mathfrak{k}^*\times \mathfrak{g}^{*n}\\
(x_0,x_1,...,x_n,g_0,g_1,...,g_n)&\mapsto (\pi_{\mathfrak{k}}(x_0),x_1,...,x_n)
\\
\mu_R:T^*(G^{n+1})&\to \mathfrak{g}^{*n}\times\mathfrak{n}^*\\
(x_0,x_1,...,x_n,g_0,g_1,...,g_n)&\mapsto (-\text{Ad}_{g_0^{-1}}^*(x_0),...,-\text{Ad}_{g_{n-1}^{-1}}^*(x_{n-1}),-\pi_{\mathfrak{n}}(\text{Ad}_{g_n^{-1}}^*x_n))
\end{aligned}
$$

We use the following two maps to model our superintegrable systems:
\begin{prop}
The maps
$$
p_1=\mu_L\times\mu_R: T^*(G^{n+1})/\mathbf{G}\to P/\mathbf{G}
$$
and
$$
p_2:P/\mathbf{G}\to (\mathfrak{g}^*/G)^{n+1},\ \ (a,x_1,...,x_n,y_1,...,y_n,b)\mapsto (-y_1,x_1,...,x_n)
$$
are well-defined Poisson maps, here we denote
$$
P:=\{(a,x_1,...,x_n,y_1,...,y_n,b)\in\mathfrak{k}^*\times \mathfrak{g}^{*n}\times\mathfrak{g}^{*n}\times\mathfrak{n}^*:a\in-\pi_\mathfrak{k}(\mathcal{O}_{y_1}),b\in-\pi_{\mathfrak{n}}(\mathcal{O}_{x_n}),x_i\in-\mathcal{O}_{y_{i+1}},\forall i\}
$$
and $\mathcal{O}_x$ denotes the coadjoint orbit of $x\in\mathfrak{g}^*$
\end{prop}

The proof is immediate.

\begin{defn}
The phase spaces of spin Toda chains are parametrized by coadjoint orbits:
$$
\mathcal{O}_K\in\mathfrak{k}^*/K,\ \mathcal{O}_N\in\mathfrak{n}^*/N,\ \mathcal{O}_1,...,\mathcal{O}_n\in\mathfrak{g}^*/G
$$
and are given by
$$
S(\mathcal{O}_K,\mathcal{O}_N,\mathcal{O}_1,...,\mathcal{O}_n):=\mu^{-1}(\mathcal{O}_K\times\mathcal{O}_1\times...\mathcal{O}_n\times\mathcal{O}_N)/\mathbf{G}
$$
\end{defn}

The rest of this section is devoted to the proof of the following theorem:
\begin{thm}
Let $P(\mathcal{O}_K,\mathcal{O}_N,\mathcal{O}_1,...,\mathcal{O}_n)$ and $B(\mathcal{O}_K,\mathcal{O}_N,\mathcal{O}_1,...,\mathcal{O}_n)$ be the images of \ \\ \ \  $S(\mathcal{O}_K,\mathcal{O}_N,\mathcal{O}_1,...,\mathcal{O}_n)$ under $p_1$ and $p_2\circ p_1$, respectively. Then the triple $(S,P,B)$ forms a superintegrable system.
\end{thm}

To prove this theorem, it remains to verify the balancing condition for dimensions. In order to compute the dimensions, we make use of the following description of the geometry of $S$.

\begin{prop}
There is an isomorphism of symplectic manifolds:
$$
S(\mathcal{O}_K,\mathcal{O}_N,\mathcal{O}_1,...,\mathcal{O}_n)\cong T^*A\times\mathcal{O}_K\times\mathcal{O}_1\times...\times\mathcal{O}_n\times\mathcal{O}_N
$$
\end{prop}

\begin{proof}
For any $1\leq i\leq n$, we have an isomorphism of Poisson manifolds:
$$
\begin{aligned}
\pi_i: (\mathfrak{g}^{*n+1}\times G^{n+1})/\mathbf{G}&\cong K\backslash(\mathfrak{g}^{*n+1}\times G)/N\\
(x_0,...,x_n,g_0,...,g_n)&\mapsto(x_0,\text{Ad}_{g_0}^*(x_1),...,\text{Ad}_{g_0...g_{i-1}}^*(x_i),\text{Ad}_{g_n^{-1}...g_{i+1}^{-1}}^*(x_{i+1}),...,\text{Ad}_{g_n^{-1}}(x_n),g_0g_1...g_n)
\end{aligned}
$$
Furthermore, using the Iwasawa decomposition $G=KAN$, we get
$$
K\backslash(\mathfrak{g}^{*n+1}\times G)/N\cong\mathfrak{g}^{*n+1}\times A
$$
Under the isomorphism $\pi_n$, $S(\mathcal{O}_K,\mathcal{O}_N,\mathcal{O}_1,...,\mathcal{O}_n)$ becomes
$$
\begin{aligned}
&S(\mathcal{O}_K,\mathcal{O}_N,\mathcal{O}_1,...,\mathcal{O}_n)\\
\cong&\{(x^{0},...,x^{n},a)\in\mathfrak{g}^{*n+1}\times A:\pi_{\mathfrak{k}}(x^{0})\in\mathcal{O}_K,-\pi_{\mathfrak{n}}(\text{Ad}_{a^{-1}}^*(x^{n}))\in\mathcal{O}_N,x^{1}-x^{0}\in\mathcal{O}_1,...,x^n-x^{n-1}\in\mathcal{O}_n\}
\end{aligned}
$$
Set
$$
\mu^K=\pi_\mathfrak{k}(x^0),\mu^N=-\pi_{\mathfrak{n}}(\text{Ad}_{a^{-1}}^*(x^{n}))
$$
$$
\mu^1=x^1-x^0,...,\mu^n=x^n-x^{n-1}
$$
For $\mu\in\mathfrak{g}$, denote by $\mu_0$ its Cartan component and by $\mu_\alpha$, $\alpha\in\Delta$ its root components, with respect to the decomposition
$$
\mathfrak{g}=\mathfrak{h}\oplus\bigoplus_{\alpha\in\Delta}\mathfrak{g}_\alpha
$$
Under the isomorphism $\mathfrak{g}^*\cong\mathfrak{g}$ induced by Killing form, we have natural identifications $\mathfrak{k}^*\cong\mathfrak{k}$ and $\mathfrak{n}^*\cong\mathfrak{n}_-$. In particular, denote
$$
\mu^K=2\sum_{\alpha\in\Delta^+}\mu_\alpha^Kk_\alpha,\ \ \ \mu^N=\sum_{\alpha\in\Delta^+}\mu^N_\alpha f_\alpha
$$
we have
$$
x_\alpha^0-x_{-\alpha}^0=\mu_\alpha^K,\ \ \ -a^\alpha x_{-\alpha}^n=\mu_\alpha^N\ \ \ \forall\alpha\in\Delta^+
$$
$$
x_\alpha^1-x_\alpha^0=\mu_\alpha^1,\ ...\ x_\alpha^n-x_\alpha^{n-1}=\mu_\alpha^n\ \ \ \forall\alpha\in\Delta
$$
from which we can obtain: for $\alpha\in\Delta^+$,
$$
x_\alpha^i=\mu_\alpha^K-a^{-\alpha}\mu_\alpha^N-\sum_{j=1}^n\mu_{-\alpha}^j+\sum_{j=1}^i\mu_\alpha^j
$$
$$
x_{-\alpha}^i=-a^{-\alpha}\mu_\alpha^N-\sum_{j=i+1}^n\mu_{-\alpha}^j
$$
For the Cartan components, set $p:=x_0^n$. Then clearly
$$
\forall 0\leq i\leq n,\ \ x^i_0=p-\sum_{j=i+1}^n\mu_0^j
$$
Therefore, $(x^0,...,x^n,a)$ and $(p,a,\mu^K,\mu^1,...,\mu^n,\mu^N)$ are in one-to-one correspondence. Regarding $(p,a)\in\mathfrak{h}^*\times A\cong T^*A$, we have proved
$$
S(\mathcal{O}_K,\mathcal{O}_N,\mathcal{O}_1,...,\mathcal{O}_n)\cong T^*A\times\mathcal{O}_K\times\mathcal{O}_1\times...\times\mathcal{O}_n\times\mathcal{O}_N
$$
\end{proof}

It follows from this proposition that
$$
\text{dim}S(\mathcal{O}_K,\mathcal{O}_N,\mathcal{O}_1,...,\mathcal{O}_n)=2r+\text{dim}(\mathcal{O}_K)+\sum_{j=1}^n\text{dim}(\mathcal{O}_j)+\text{dim}(\mathcal{O}_N)
$$

By definition, we have:
$$
\begin{aligned}
&B(\mathcal{O}_K,\mathcal{O}_N,\mathcal{O}_1,...,\mathcal{O}_n)\\
=&\{(\mathcal{O}^0,...,\mathcal{O}^n)\in(\mathfrak{g}^*/G)^{n+1}:\mathcal{O}_K\subset\pi_{\mathfrak{k}}(\mathcal{O}^0),\mathcal{O}_N\subset-\pi_{\mathfrak{n}}(\mathcal{O}^n), \mathcal{O}_1\in\mathcal{O}^1-\mathcal{O}^0,...,\mathcal{O}_n\in\mathcal{O}^n-\mathcal{O}^{n-1}\}
\end{aligned}
$$
while the fiber of $p_2$ is given by
$$
p_2^{-1}(\mathcal{O}^0,...,\mathcal{O}^n)\cong\left( (\mathcal{O}^0\cap\pi_{\mathfrak{k}}^{-1}(\mathcal{O}_K))/K\right)\times\prod_{i=1}^n\mathcal{M}(-\mathcal{O}^i,\mathcal{O}^{i-1},\mathcal{O}_i)\times\left( (\mathcal{O}^n\cap\pi_{\mathfrak{n}}^{-1}(-\mathcal{O}_N))/N\right)
$$
Here, for any coadjoint orbits $\mathcal{O}_1,\mathcal{O}_2,\mathcal{O}_3\subset\mathfrak{g}^*$, we define
$$
\mathcal{M}(\mathcal{O}_1,\mathcal{O}_2,\mathcal{O}_3):=\{(x,y,z)\in\mathcal{O}_1\times\mathcal{O}_2\times\mathcal{O}_3:x+y+z=0\}/G
$$
We recall the following results from \cite{res}
\begin{prop}
For generic sufficiently large coadjoint orbits $\mathcal{O}_1,\mathcal{O}_2$ and a non-zero orbit $\mathcal{O}_3$, we have:
$$
\text{dim}\mathcal{M}(\mathcal{O}_1,\mathcal{O}_2,\mathcal{O}_3)=\text{dim}(\mathcal{O}_3)-2r
$$
\end{prop}

\begin{prop}
For generic coadjoint orbits, the dimension of the base space is
$$
\text{dim}B(\mathcal{O}_K,\mathcal{O}_N,\mathcal{O}_1,...,\mathcal{O}_n)=(n+1)r
$$
\end{prop}

We can now prove the superintegrability of spin Toda chains.

\begin{proof}[Proof of the theorem]
For generic $\mathcal{O}_K$, the linear projection $\pi_{\mathfrak{k}}$ is transversal to the level sets of the invariant polynomials, hence
$$
\text{dim}(\mathcal{O}\cap\pi_{\mathfrak{k}}^{-1}(\mathcal{O}_K))/K=\text{dim}(\mathcal{O}_K)
$$
For the same reason:
$$
\text{dim}(\mathcal{O}\cap\pi_{\mathfrak{n}}^{-1}(-\mathcal{O}_N))/N=\text{dim}(\mathcal{O}_N)
$$
We get
$$
\begin{aligned}
\text{dim}P(\mathcal{O}_K,\mathcal{O}_N,\mathcal{O}_1,...,\mathcal{O}_n)=&\text{dim}B(\mathcal{O}_K,\mathcal{O}_N,\mathcal{O}_1,...,\mathcal{O}_n)+\text{dim}(\text{generic fiber})\\
=&(n+1)r+\text{dim}(\mathcal{O}\cap\pi_{\mathfrak{k}}^{-1}(\mathcal{O}_K))/K\\
&+\sum_{i=1}^n\text{dim}\mathcal{M}(-\mathcal{O}^i,\mathcal{O}^{i-1},\mathcal{O}_i)+\text{dim}(\mathcal{O}\cap\pi_{\mathfrak{n}}^{-1}(-\mathcal{O}_N))/N\\
=&(n+1)r+\text{dim}(\mathcal{O}_K)+\sum_{i=1}^n(\text{dim}(\mathcal{O}_i)-2r)+\text{dim}(\mathcal{O}_N)\\
=&(1-n)r+\text{dim}(\mathcal{O}_K)+\sum_{i=1}^n\text{dim}(\mathcal{O}_i)+\text{dim}(\mathcal{O}_N)\\
=&\text{dim}S(\mathcal{O}_K,\mathcal{O}_N,\mathcal{O}_1,...,\mathcal{O}_n)-\text{dim}B(\mathcal{O}_K,\mathcal{O}_N,\mathcal{O}_1,...,\mathcal{O}_n)
\end{aligned}
$$
Therefore, the spin Toda chain is superintegrable.
\end{proof}

By the general theory of superintegrable systems, any function on the base manifold $B(\mathcal{O}_K,\mathcal{O}_N,\mathcal{O}_1,...,\mathcal{O}_n)$, that is, any invariant function on any factor of $(\mathfrak{g}^*/G)^{n+1}$, will give rise to Hamiltonians for the spin Toda system. In particular, consider the quadratic Casimir functions:
$$
H_2^k(x)=\frac{1}{2}(x^k,x^k)=\frac{1}{2}(x_0^k,x_0^k)+\sum_{\alpha>0}x_\alpha^kx_{-\alpha}^k,\ \ \ 0\leq k\leq n
$$

From these, we obtain the second-order Hamiltonian $H_2^n$ and first-order functions $D_k=H_2^k-H_2^{k-1}$. Expressing them in the coordinates of $T^*A\times\mathcal{O}_K\times\mathcal{O}_1\times...\times\mathcal{O}_n\times\mathcal{O}_N$, we get
$$
\begin{aligned}
H_2^n&=\frac{1}{2}(p,p)+\sum_{\alpha>0}x_\alpha^nx_{-\alpha}^n\\
&=\frac{1}{2}(p,p)+\sum_{\alpha>0}\left(\mu_{\alpha}^K-a^{-\alpha}\mu_\alpha^N+\sum_{j=1}^n(\mu_\alpha^j-\mu_{-\alpha}^j)\right)(-a^{-\alpha}\mu_\alpha^N)\\
&=\frac{1}{2}(p,p)+\sum_{\alpha>0}a^{-2\alpha}(\mu_\alpha^N)^2-\sum_{\alpha>0}a^{-\alpha}\left(\mu_{\alpha}^K+\sum_{j=1}^n(\mu_\alpha^j-\mu_{-\alpha}^j)\right)\mu_\alpha^N
\end{aligned}
$$
which is clearly the classical counterpart of the spin Toda Hamiltonian $\mathcal{H}^{\sigma^{(n)}}$, and
$$
\begin{aligned}
D_k=&\frac{1}{2}(x^k,x^k)-\frac{1}{2}(x^{k-1},x^{k-1})\\
=&\frac{1}{2}(\mu^k+x^{k-1},\mu^k+x^{k-1})-\frac{1}{2}(x^{k-1},x^{k-1})\\
=&(\mu^k,x^{k-1})+\frac{1}{2}(\mu^k,\mu^k)\\
=&(\mu_0^k,x_0^{k-1})+\sum_{\alpha>0}(x_\alpha^{k-1}\mu_{-\alpha}^k+x_{-\alpha}^{k-1}\mu_\alpha^k)+\frac{1}{2}(\mu_0^k,\mu_0^k)+\sum_{\alpha>0}\mu_\alpha^k\mu_{-\alpha}^k\\
=&(\mu_0^k,p-\sum_{j=k}^n\mu_0^j)+\frac{1}{2}(\mu_0^k,\mu_0^k)\\
&+\sum_{\alpha>0}\left((\mu_\alpha^K-a^{-\alpha}\mu_\alpha^N-\sum_{j=1}^n\mu_{-\alpha}^j+\sum_{j=1}^{k-1}\mu_\alpha^j)\mu_{-\alpha}^k+(-a^{-\alpha}\mu_\alpha^N-\sum_{j=k}^n\mu_{-\alpha}^j)\mu_\alpha^k+\mu_\alpha^k\mu_{-\alpha}^k\right)\\
=&(\mu_0^k,p)-\frac{1}{2}(\mu_0^k,\mu_0^k)-\sum_{j=k+1}^n(\mu_0^k,\mu_0^j)+\sum_{j=1}^{k-1}\sum_{\alpha>0}(\mu_\alpha^j-\mu_{-\alpha}^j)\mu_{-\alpha}^k-\sum_{j=k+1}^n\sum_{\alpha>0}\mu_{-\alpha}^j(\mu_\alpha^k+\mu_{-\alpha}^k)\\
&+\sum_{\alpha>0}(\mu_\alpha^K\mu_{-\alpha}^k-a^{-\alpha}\mu_\alpha^N(\mu_\alpha^k+\mu_{-\alpha}^k))-\sum_{\alpha>0}(\mu_{-\alpha}^k)^2\\
=&(\mu_0^k,p)+(r^K_k-r^N_k-\kappa_k)+\sum_{j=1}^{k-1}r^l_{jk}-\sum_{j=k+1}^n r^r_{kj}
\end{aligned}
$$

Here
$$
r^K_k=\sum_{\alpha>0}\mu_\alpha^K\mu_{-\alpha}^k,\ \ \ r^N_k=\sum_{\alpha>0}a^{-\alpha}\mu_\alpha^N(\mu_\alpha^k+\mu_{-\alpha}^k)
$$
$$
\kappa_k=\frac{1}{2}(\mu_0^k,\mu_0^k)+\sum_{\alpha>0}(\mu_{-\alpha}^k)^2
$$
$$
r^l_{jk}=\sum_{\alpha>0}(\mu_\alpha^j-\mu_{-\alpha}^j)\mu_{-\alpha}^k
$$
$$
r^r_{kj}=(\mu_0^k,\mu_0^j)+\sum_{\alpha>0}\mu_{-\alpha}^j(\mu_\alpha^k+\mu_{-\alpha}^k)
$$
This is the classical counterpart of the boundary KZB type operators and of the factorisation.

\begin{rem}
By a method entirely analogous to that in \cite{rsh}, the calculations for the classical case presented here can be updated to the quantum case to construct the boundary KZB-type operators. These operators now arise from the differences of elements in $Z(U(\mathfrak{g}))^{\otimes(N+1)}$, acting naturally on the space of functions
$$
\begin{aligned}
&C^\infty((K,\sigma)\backslash G/(G,\pi_{U_1})\backslash G/...\backslash G/(G,\pi_{U_N})\backslash G/(N,\chi^*))\\
=&\Bigl\{f\in C^\infty(G^{N+1})\otimes V\otimes\mathbf{U}: f(kg_0h_1^{-1},h_1g_1h_2^{-1},...,h_Ng_N n^{-1})\\
&\ \ \ \ \ \ \ \ \ \ \ \ \ \ \ \ \ \ \ \ \ \ \ \ \ \ \ \ \ \ \ \ =\chi(n)^{-1}(\sigma(k)\otimes\pi_{U_1}(g_1)\otimes...\otimes\pi_{U_N}(g_N))f(g_0,...,g_N)\Bigr\}\\
\cong& C^{\infty}((K,\sigma^{(N)})\backslash G/(N,\chi^*))
\end{aligned}
$$
Here $\mathbf{U}=U_1\otimes...\otimes U_N$. In particular, this establishes the commutativity of boundary KZB-type operators:
$$
[\tilde{D}_i,\tilde{D}_j]=0,\ \ \ [\tilde{D}_i,\tilde{H}^{(\sigma^(N),\chi^*)}]=0
$$

Furthermore, N-point spherical-Whittaker functions may be characterized among elementary spherical-Whittaker functions (that is, spherical Whittaker functions of the form $f(g)=\phi_l\circ\pi(g)\phi_r$, with $\phi_l\in\text{Hom}_K(\mathcal{H}_\lambda,V_{\sigma^{(N)}})$ and $\phi_r\in\text{Hom}_N(\mathbbm{C}_\chi,\mathcal{H}_\lambda^{-\infty})$ for some principal series representations.), as the joint eigenfunctions of the commuting boundary KZB-type operators. This phenomenon, together with a more precise description of the wave functions for a broader class of quasi-simple representations, will be investigated further in future work.
\end{rem}

\end{document}